\documentclass[journal]{IEEEtran}

\usepackage[utf8]{inputenc}
\usepackage{color}
\usepackage{array}
\usepackage{verbatim}
\usepackage{float}
\usepackage{amsmath}
\usepackage{amsthm}
\usepackage{amssymb}
\usepackage{graphicx}
\usepackage{longtable}
\usepackage{multirow}
\usepackage[unicode=true,
bookmarks=false,
breaklinks=false,pdfborder={0 0 1},colorlinks=false]
{hyperref}
\hypersetup{
	colorlinks,bookmarksopen,bookmarksnumbered,citecolor=blue,urlcolor=blue}
\usepackage{cite}

\floatstyle{ruled}
\newfloat{algorithm}{tbp}{loa}
\providecommand{\algorithmname}{Algorithm}
\floatname{algorithm}{\protect\algorithmname}

\makeatletter
\let\oldforeign@language\foreign@language
\DeclareRobustCommand{\foreign@language}[1]{%
	\lowercase{\oldforeign@language{#1}}}

\let\oldforeign@language\foreign@language
\DeclareRobustCommand{\foreign@language}[1]{%
	\lowercase{\oldforeign@language{#1}}}

\ifCLASSINFOpdf
\else
\fi

\newcommand{\MYfooter}{\smash{
		\hfil\parbox[t][\height][t]{\textwidth}{\centering
			\thepage}\hfil\hbox{}}}

\makeatletter

\def\ps@IEEEtitlepagestyle{%
	\def\@oddhead{\parbox[t][\height][t]{\textwidth}{\centering \scriptsize
			Personal use of this material is permitted. Permission from the author(s) and/or copyright holder(s), must be obtained for all other uses. Please contact us and provide details if you believe this document breaches copyrights.\\
			\noindent\makebox[\linewidth]{}
		}\hfil\hbox{}}%
	\def\@evenhead{\scriptsize\thepage \hfil \leftmark\mbox{}}%
	\def\@oddfoot{\parbox[t][\height][l]{\textwidth}{
			\vspace{-20pt}{\rule{\textwidth}{0.4pt}}\\ \footnotesize\underline{To cite this article:}
			{\bf{\textcolor{red}{H. A. Hashim, "Guaranteed Performance Nonlinear Observer for Simultaneous Localization and Mapping," IEEE Control Systems Letters, vol. 5, no. 1, pp. 91-96, 2021.}}} doi: \href{https://doi.org/10.1109/LCSYS.2020.3000266}{10.1109/LCSYS.2020.3000266}\\
			\noindent\makebox[\linewidth]
		}\hfil\hbox{}}%
	\def\@evenfoot{\MYfooter}}

\makeatother
\newtheorem{thm}{Theorem}
\newtheorem{rem}{Remark}

\newtheorem{assum}{Assumption}

\begin{document}
	\bstctlcite{IEEEexample:BSTcontrol}

\title{Guaranteed Performance Nonlinear Observer for Simultaneous Localization and Mapping}

\author{Hashim A. Hashim
	\thanks{This work was supported in part by Thompson Rivers University Internal research fund, RGS-2020/21 IRF, \# 102315.}
	\thanks{Corresponding author, H. A. Hashim is with the Department of Engineering and Applied Science, Thompson Rivers University, Kamloops, British Columbia, Canada, V2C-0C8, e-mail: hhashim@tru.ca.}
	
}

\markboth{}{Hashim \MakeLowercase{\textit{et al.}}: Guaranteed Performance Nonlinear Observer for SLAM}

\maketitle

\begin{abstract}
A geometric nonlinear observer algorithm for Simultaneous Localization
and Mapping (SLAM) developed on the Lie group of $\mathbb{SLAM}_{n}\left(3\right)$
is proposed. The presented novel solution estimates the vehicle's
pose (\textit{i.e}. attitude and position) with respect to landmarks
simultaneously positioning the reference features in the global frame.
The proposed estimator on manifold is characterized by predefined
measures of transient and steady-state performance. Dynamically reducing
boundaries guide the error function of the system to reduce asymptotically
to the origin from its starting position within a large given set.
The proposed observer has the ability to use the available velocity
and feature measurements directly. Also, it compensates for unknown
constant bias attached to velocity measurements. Unit-qauternion of
the proposed observer is presented. Numerical results reveal effectiveness
of the proposed observer.
\end{abstract}


\IEEEpeerreviewmaketitle{}

\section{Introduction}

\IEEEPARstart{N}{avigation} solutions, in the age of autonomous vehicles, suitable
for both partially and entirely unknown environments are an absolute
necessity. Autonomous navigation systems are an integral part of a
variety of applications including household autonomous devices, mine
exploration, location of missing terrestrial, underwater vehicles
and others. The nature of these applications limits the usefulness
of absolute positioning systems, such as global positioning systems
(GPS) which require visibility of at least four satellites. In the
absence of GPS, other techniques are used. If pose of a robot or vehicle
is known, while the map of its surroundings is unknown, the problem
is referred to as a mapping problem \cite{thrun2002robotic}. On the
contrary, if the map of the environment is known, while the pose is
unknown, the problem is described as pose estimation \cite{hashim2019SE3Det,hua2015gradient,hashim2020SE3Stochastic,hashim2018SE3Stochastic}.
Simultaneous Localization and Mapping (SLAM) combines mapping and
pose estimation problems and requires the autonomous system to simultaneously
build a map of the environment and track its own pose (\textit{i.e}.
attitude and position) within that environment. SLAM problem can be
solved using a set of measurements available at the body-fixed frame
of the vehicle.

Over the last few decades, Gaussian filters played a significant role
in solving the SLAM problem by positioning both the vehicle and its
surrounding features. Commonly used algorithms include FastSLAM \cite{montemerlo2007fastslam},
incremental SLAM \cite{kaess2008isam}, particle filter \cite{bekris2006evaluation},
and invariant EKF \cite{zhang2017EKF_SLAM}. The SLAM algorithms proposed
in \cite{montemerlo2007fastslam,kaess2008isam,bekris2006evaluation,huang2007convergence,zhang2017EKF_SLAM}
are based on probabilistic approach. Over a decade ago, graphical
maximum likelihood algorithms have been widely explored \cite{grisetti2010tutorial,cadena2016past}.
Aside from Gaussian filtering methods, the true SLAM problem is a
dual estimation problem which is highly nonlinear in nature, and evolves
directly on the Lie group of $\mathbb{SLAM}_{n}\left(3\right)$ which
will be defined in the next Section. For instance, the pose dynamics
are modeled on the Lie group of the special Euclidean group $\mathbb{SE}\left(3\right)$.
Over the last ten years, several nonlinear observers developed directly
on $\mathbb{SE}\left(3\right)$ have been proposed, for instance \cite{hashim2019SE3Det,hashim2020SE3Stochastic}.
As a result, manifolds and the inheritance of the Lie group of $\mathbb{SE}\left(3\right)$
in the SLAM problem was studied \cite{strasdat2012local,forster2016manifold}.
A two stage observer for SLAM has been presented in \cite{johansen2016globally}
where the first stage consists of a nonlinear pose observer, while
the second stage is comprised of a Kalman filter used for feature
estimation. Nevertheless, the above approach did not capture the true
nonlinearity of the SLAM problem. Although the nonlinear observers
proposed in \cite{mahony2017SLAM,zlotnik2018SLAM} mimic the nonlinear
structure of the true SLAM problem, they lack measures of the error
convergence for the transient and steady-state performance. 

This work introduces a novel nonlinear observer evolved directly on
the Lie group of $\mathbb{SLAM}_{n}\left(3\right)$ using velocity
and feature measurements. In view of practical implementation and
similar to \cite{zlotnik2018SLAM}, the velocity measurements are
assumed to be corrupted with unknown bias. With the aim of achieving
systematic convergence of the SLAM error function, the error is constrained
to initiate among a predefined known large set and reduce systematically
and smoothly obeying predefined dynamically reducing boundaries and
to settle within a known small set, unlike to \cite{zlotnik2018SLAM}.
Prescribed performance function (PPF) captures the concept of systematic
convergence \cite{bechlioulis2008robust}. PPF forces the error to
be constrained by introducing a new form of unconstrained error, termed
transformed error. %
The nonlinear observer is designed such that the SLAM error function
as well as the transformed error can be proven to be globally asymptotically
stable. 

The Introduction section is followed by five sections, where Section
\ref{sec:Preliminaries-and-Math} overviews mathematical notation,
Lie group of $\mathbb{SE}\left(3\right)$, and $\mathbb{SLAM}_{n}\left(3\right)$.
Section \ref{sec:SE3_Problem-Formulation} introduces the SLAM problem
along with available measurements. Section \ref{sec:SLAM_Filter}
reformulates the SLAM problem to satisfy PPF and presents a nonlinear
observer design on $\mathbb{SLAM}_{n}\left(3\right)$ with systematic
convergence. Section \ref{sec:SE3_Simulations} includes simulation
results. Finally, Section \ref{sec:SE3_Conclusion} concludes the
work.

\section{Preliminaries and Math Notation \label{sec:Preliminaries-and-Math}}

Consider a vehicle traveling in three dimensional (3D) space. The
vehicle fixed body-frame is described by $\left\{ \mathcal{B}\right\} $
and the absolute fixed inertial-frame is described by $\left\{ \mathcal{I}\right\} $.
The set of real numbers, nonnegative real numbers, and real $n$-by-$m$
space, are defined by $\mathbb{R}$, $\mathbb{R}_{+}$, and $\mathbb{R}^{n\times m}$,
respectively. $\mathbf{I}_{n}$ refers to $n$-dimensional identity
matrix, $\underline{\mathbf{0}}_{n}$ describes a zero column vector.
For $x\in\mathbb{R}^{n}$ the Euclidean norm is $\left\Vert x\right\Vert =\sqrt{x^{\top}x}$.
Vehicle attitude is described by $R\in\mathbb{SO}\left(3\right)$
where $\mathbb{SO}\left(3\right)$ denotes Special Orthogonal Group
such that $\mathbb{SO}\left(3\right)=\{\left.R\in\mathbb{R}^{3\times3}\right|RR^{\top}=\mathbf{I}_{3}\text{, }{\rm det}\left(R\right)=+1\}$
with ${\rm det\left(\cdot\right)}$ representing a determinant, visit
\cite{hashim2019SO3Wiley,hashim2018SO3Stochastic}. $\boldsymbol{T}\in\mathbb{R}^{4\times4}$
describes the vehicle's pose in 3D space expressed as 
\begin{equation}
\boldsymbol{T}=\left[\begin{array}{cc}
R & P\\
\underline{\mathbf{0}}_{3}^{\top} & 1
\end{array}\right]\in\mathbb{SE}\left(3\right)\label{eq:T_SLAM}
\end{equation}
where $P\in\mathbb{R}^{3}$ refers to the vehicle's position, $R\in\mathbb{SO}\left(3\right)$
defines vehicle's attitude, and $\mathbb{SE}\left(3\right)$ refers
to Special Euclidean Group described by $\mathbb{SE}\left(3\right)=\{\left.\boldsymbol{T}\in\mathbb{R}^{4\times4}\right|R\in\mathbb{SO}\left(3\right),P\in\mathbb{R}^{3}\}$,
visit \cite{hashim2019SE3Det}. $\mathfrak{so}\left(3\right)$ is
the Lie-algebra of $\mathbb{SO}\left(3\right)$ described by $\mathfrak{so}\left(3\right)=\{\left.\left[y\right]_{\times}\in\mathbb{R}^{3\times3}\right|\left[y\right]_{\times}^{\top}=-\left[y\right]_{\times},y\in\mathbb{R}^{3}\}$
with $\left[y\right]_{\times}$ denoting a skew symmetric matrix
\[
\left[y\right]_{\times}=\left[\begin{array}{ccc}
0 & -y_{3} & y_{2}\\
y_{3} & 0 & -y_{1}\\
-y_{2} & y_{1} & 0
\end{array}\right]\in\mathfrak{so}\left(3\right),\hspace{1em}y=\left[\begin{array}{c}
y_{1}\\
y_{2}\\
y_{3}
\end{array}\right]
\]
$\mathfrak{se}\left(3\right)$ is the Lie-algebra of $\mathbb{SE}\left(3\right)$
with {\small{}
	\[
	\mathfrak{se}\left(3\right)=\left\{ \left[U\right]_{\wedge}\in\mathbb{R}^{4\times4}\left|\exists\Omega,V\in\mathbb{R}^{3}:\left[U\right]_{\wedge}=\left[\begin{array}{cc}
	\left[\Omega\right]_{\times} & V\\
	\underline{\mathbf{0}}_{3}^{\top} & 0
	\end{array}\right]\right.\right\} 
	\]
	where} $U=\left[\Omega^{\top},V^{\top}\right]\in\mathbb{R}^{6}$.
Consider that $\boldsymbol{T}\in\mathbb{SE}\left(3\right)$ as defined
in \eqref{eq:T_SLAM} and $U\in\mathbb{R}^{6}$. Define the adjoint
map ${\rm Ad}_{\boldsymbol{T}}:\mathbb{SE}\left(3\right)\times\mathfrak{se}\left(3\right)\rightarrow\mathfrak{se}\left(3\right)$
and the augmented adjoint map $\overline{{\rm Ad}}_{\boldsymbol{T}}:\mathbb{SE}\left(3\right)\rightarrow\mathbb{R}^{6\times6}$
as below
\begin{equation}
\begin{cases}
{\rm Ad}_{\boldsymbol{T}}\left(\left[U\right]_{\wedge}\right) & =\boldsymbol{T}\left[U\right]_{\wedge}\boldsymbol{T}^{-1}\in\mathfrak{se}\left(3\right)\\
\overline{{\rm Ad}}_{\boldsymbol{T}} & =\left[\begin{array}{cc}
R & 0_{3\times3}\\
\left[P\right]_{\times}R & R
\end{array}\right]\in\mathbb{R}^{6\times6}
\end{cases}\label{eq:SLAM_Adjoint}
\end{equation}
In view of \eqref{eq:SLAM_Adjoint}, one finds
\begin{equation}
{\rm Ad}_{\boldsymbol{T}}\left(\left[U\right]_{\wedge}\right)=\left[\,\overline{{\rm Ad}}_{\boldsymbol{T}}U\right]_{\wedge},\hspace{1em}\boldsymbol{T}\in\mathbb{SE}\left(3\right),U\in\mathbb{R}^{6}\label{eq:SLAM_Adjoint_MAP}
\end{equation}
Define the sub-manifolds $\overset{\circ}{\mathcal{M}}$ and $\overline{\mathcal{M}}$
of $\mathbb{R}^{4}$ as
\begin{align*}
\overset{\circ}{\mathcal{M}} & =\left\{ \left.\overset{\circ}{y}=\left[\begin{array}{cc}
y^{\top} & 0\end{array}\right]^{\top}\in\mathbb{R}^{4}\right|y\in\mathbb{R}^{3}\right\} \\
\overline{\mathcal{M}} & =\left\{ \left.\overline{y}=\left[\begin{array}{cc}
y^{\top} & 1\end{array}\right]^{\top}\in\mathbb{R}^{4}\right|y\in\mathbb{R}^{3}\right\} 
\end{align*}
Let the Lie group of $\mathbb{SLAM}_{n}\left(3\right)=\mathbb{SE}\left(3\right)\times\overline{\mathcal{M}}^{n}$
be
\begin{equation}
\mathbb{SLAM}_{n}\left(3\right)=\left\{ X=\left(\boldsymbol{T},\overline{{\rm p}}\right)\left|\boldsymbol{T}\in\mathbb{SE}\left(3\right),\overline{{\rm p}}\in\overline{\mathcal{M}}^{n}\right.\right\} \label{eq:SLAM_SLAM_X}
\end{equation}
where $\overline{{\rm p}}=\left[\overline{{\rm p}}_{1},\overline{{\rm p}}_{2},\ldots,\overline{{\rm p}}_{n}\right]\in\overline{\mathcal{M}}^{n}$
and $\overline{\mathcal{M}}^{n}=\overline{\mathcal{M}}\times\overline{\mathcal{M}}\times\cdots\times\overline{\mathcal{M}}$.
Describe the tangent space at the identity element of $X=\left(\boldsymbol{T},\overline{{\rm p}}\right)\in\mathbb{SLAM}_{n}\left(3\right)$
as $\mathfrak{slam}_{n}\left(3\right)=\mathfrak{se}\left(3\right)\times\overset{\circ}{\mathcal{M}}^{n}$
\begin{equation}
\mathfrak{slam}_{n}\left(3\right)=\left\{ \mathcal{Y}=(\left[U\right]_{\wedge},\overset{\circ}{{\rm v}})\left|\left[U\right]_{\wedge}\in\mathfrak{se}\left(3\right),\overset{\circ}{{\rm v}}\in\overset{\circ}{\mathcal{M}}^{n}\right.\right\} \label{eq:SLAM_SLAM_Y}
\end{equation}
where $\overset{\circ}{{\rm v}}=[\overset{\circ}{{\rm v}}_{1},\overset{\circ}{{\rm v}}_{2},\ldots,\overset{\circ}{{\rm v}}_{n}]\in\overset{\circ}{\mathcal{M}}^{n}$,
$\overset{\circ}{\mathcal{M}}^{n}=\overset{\circ}{\mathcal{M}}\times\overset{\circ}{\mathcal{M}}\times\cdots\times\overset{\circ}{\mathcal{M}}$,
and $\overset{\circ}{{\rm v}}_{i}=\left[{\rm v}_{i}^{\top},0\right]^{\top}\in\overset{\circ}{\mathcal{M}}$.
Also, $\overline{{\rm p}}_{i}=\left[{\rm p}_{i}^{\top},1\right]^{\top}\in\overline{\mathcal{M}}$.

\section{Problem Formulation\label{sec:SE3_Problem-Formulation}}

Consider a vehicle moving in 3D space within a map that has $n$ features.
SLAM problem is the process of estimating vehicle pose $\boldsymbol{T}\in\mathbb{SE}\left(3\right)$,
and at the same time estimating $n$ features within the environment
$\overline{{\rm p}}=\left[\overline{{\rm p}}_{1},\overline{{\rm p}}_{2},\ldots,\overline{{\rm p}}_{n}\right]\in\overline{\mathcal{M}}^{n}$.
Fig. \ref{fig:SLAM} provides a conceptual illustration of the SLAM
estimation problem. 
\begin{figure*}
	\centering{}\includegraphics[scale=0.63]{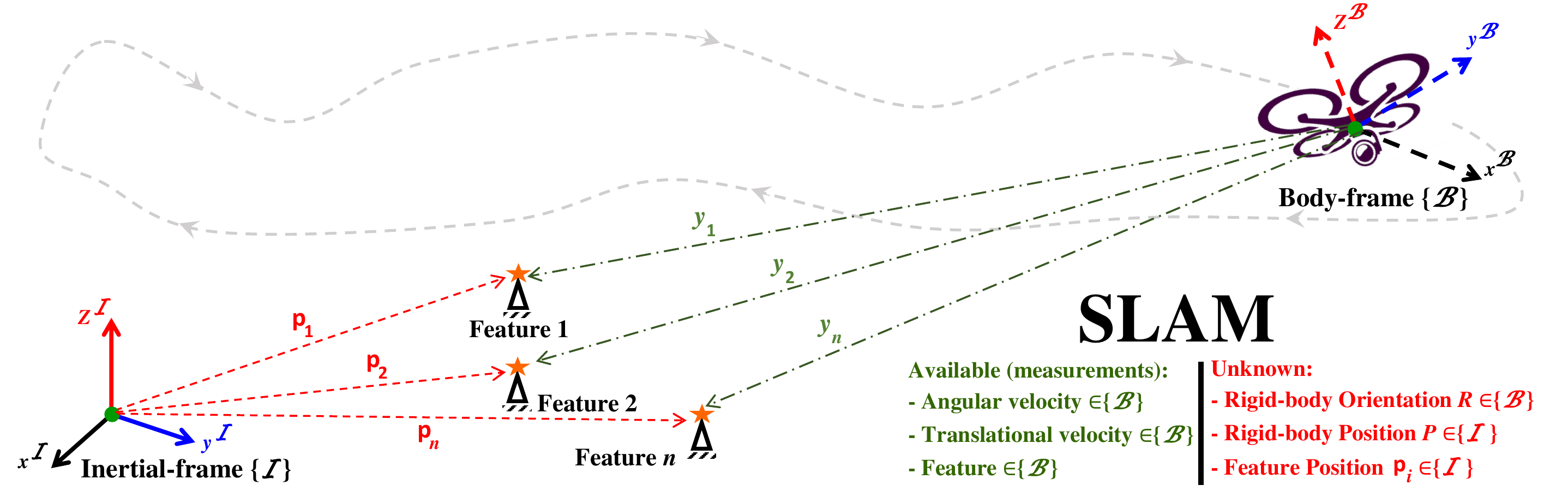}\caption{SLAM estimation problem.}
	\label{fig:SLAM}
\end{figure*}

Consider $R\in\mathbb{SO}\left(3\right)$ and $P\in\mathbb{R}^{3}$
to be vehicle's attitude (orientation) and position, respectively,
and ${\rm p}_{i}\in\mathbb{R}^{3}$ to be the $i$th feature position
in the map, where $R\in\left\{ \mathcal{B}\right\} $, and $P,{\rm p}_{i}\in\left\{ \mathcal{I}\right\} $
for all $i=1,2,\ldots,n$. Let $X=(\boldsymbol{T},\overline{{\rm p}})\in\mathbb{SLAM}_{n}\left(3\right)$
be the true pose and features with $\boldsymbol{T}\in\mathbb{SE}\left(3\right)$
and $\overline{{\rm p}}\in\overline{\mathcal{M}}^{n}$. Let $\mathcal{Y}=(\left[U\right]_{\wedge},\overset{\circ}{{\rm v}})\in\mathfrak{slam}_{n}\left(3\right)$
be the true group velocity where $\overset{\circ}{{\rm v}}=[\overset{\circ}{{\rm v}}_{1},\overset{\circ}{{\rm v}}_{2},\ldots,\overset{\circ}{{\rm v}}_{n}]\in\overset{\circ}{\mathcal{M}}^{n}$
and $U,\overset{\circ}{{\rm v}}\in\left\{ \mathcal{B}\right\} $.
The true motion dynamics of SLAM are as follows:
\begin{equation}
\begin{cases}
\dot{\boldsymbol{T}} & =\boldsymbol{T}\left[U\right]_{\wedge}\\
\dot{{\rm p}}_{i} & =R{\rm v}_{i},\hspace{1em}\forall i=1,2,\ldots,n
\end{cases}\label{eq:SLAM_True_dot}
\end{equation}
where $U=\left[\Omega^{\top},V^{\top}\right]^{\top}$, $\Omega\in\mathbb{R}^{3}$
is the true angular velocity, $V\in\mathbb{R}^{3}$ is the true translational
velocity, while ${\rm v}_{i}\in\mathbb{R}^{3}$ describes the $i$th
linear velocity of ${\rm p}_{i}$. $X$ is unknown and can be obtained
with the aid of 1) $\mathcal{Y}_{m}=\left(\left[U_{m}\right]_{\wedge},\overset{\circ}{{\rm v}}_{m}\right)\in\mathfrak{slam}_{n}\left(3\right)$
which represents velocity measurements and 2) $\overline{y}_{i}\in\overline{\mathcal{M}}$
which is the $i$th feature measurement for all $\mathcal{Y}_{m},\overline{y}_{i}\in\left\{ \mathcal{B}\right\} $.
Since features are fixed to $\left\{ \mathcal{I}\right\} $, $\dot{{\rm p}}_{i}=\underline{\mathbf{0}}_{3}$
and consequently ${\rm v}_{i}=\underline{\mathbf{0}}_{3}$. The measurement
of the group velocity $U_{m}=\left[\Omega_{m}^{\top},V_{m}^{\top}\right]^{\top}$
is 
\begin{equation}
U_{m}=U+b_{U}+n_{U}\in\mathbb{R}^{6}\label{eq:SLAM_TVelcoity}
\end{equation}
where $b_{U}=\left[b_{\Omega}^{\top},b_{V}^{\top}\right]^{\top}$
is unknown constant bias and $n_{U}$ denotes random noise. The $i$th
feature measurement in the body-frame is described by
\begin{equation}
\overline{y}_{i}=\boldsymbol{T}^{-1}\overline{{\rm p}}_{i}+\overset{\circ}{b}_{i}^{y}+\overset{\circ}{n}_{i}^{y}\in\overline{\mathcal{M}},\hspace{1em}\forall i=1,2,\ldots,n\label{eq:SLAM_Vec_Landmark}
\end{equation}
where $\overset{\circ}{b}_{i}^{y}\in\overset{\circ}{\mathcal{M}}$
and $\overset{\circ}{n}_{i}^{y}\in\overset{\circ}{\mathcal{M}}$ represent
unknown constant bias and random noise, respectively. Also, $\overline{{\rm p}}_{i}=[{\rm p}_{i}^{\top},1]^{\top}\in\overline{\mathcal{M}}$
denotes the $i$th feature. In our analysis, $n_{U}$, $b_{i}^{y}$,
and $n_{i}^{y}$ are zeros.

\begin{assum}\label{Assumption:Feature}Three or more features available
	for measurement that define a plane with $\overline{y}=\left[\overline{y}_{1},\overline{y}_{2},\ldots,\overline{y}_{n}\right]\in\overline{\mathcal{M}}^{n}$.\end{assum}

Define the estimate of pose as
\[
\hat{\boldsymbol{T}}=\left[\begin{array}{cc}
\hat{R} & \hat{P}\\
\underline{\mathbf{0}}_{3}^{\top} & 1
\end{array}\right]\in\mathbb{SE}\left(3\right)
\]
where $\hat{R}$ and $\hat{P}$ represent estimates of the true orientation
and position, respectively. Define $\hat{{\rm p}}_{i}$ as the estimate
of the true $i$th feature ${\rm p}_{i}$. Consider the error between
$\boldsymbol{T}$ and $\hat{\boldsymbol{T}}$ as
\begin{align}
\tilde{\boldsymbol{T}}=\hat{\boldsymbol{T}}\boldsymbol{T}^{-1} & =\left[\begin{array}{cc}
\hat{R} & \hat{P}\\
\underline{\mathbf{0}}_{3}^{\top} & 1
\end{array}\right]\left[\begin{array}{cc}
R^{\top} & -R^{\top}P\\
\underline{\mathbf{0}}_{3}^{\top} & 1
\end{array}\right]\nonumber \\
& =\left[\begin{array}{cc}
\tilde{R} & \tilde{P}\\
\underline{\mathbf{0}}_{3}^{\top} & 1
\end{array}\right]\label{eq:SLAM_T_error}
\end{align}
with $\tilde{R}=\hat{R}R^{\top}$and $\tilde{P}=\hat{P}-\tilde{R}P$
describing error in orientation and position, respectively. Define
the error between $\hat{{\rm p}}_{i}$ and ${\rm p}_{i}$ as
\begin{equation}
\overset{\circ}{e}_{i}=\overline{\hat{{\rm p}}}_{i}-\tilde{\boldsymbol{T}}\,\overline{{\rm p}}_{i}\label{eq:SLAM_e}
\end{equation}
where $\overset{\circ}{e}_{i}=[e_{i}^{\top},0]^{\top}\in\overset{\circ}{\mathcal{M}}$
and $\overline{\hat{{\rm p}}}_{i}=\left[\hat{{\rm p}}_{i}^{\top},1\right]^{\top}\in\overline{\mathcal{M}}$.
In the light of $\tilde{\boldsymbol{T}}$, definition in \eqref{eq:SLAM_T_error},
and \eqref{eq:SLAM_Vec_Landmark}, 
\begin{align}
\overset{\circ}{e}_{i} & =\overline{\hat{{\rm p}}}_{i}-\hat{\boldsymbol{T}}\boldsymbol{T}^{-1}\,\overline{{\rm p}}_{i}\nonumber \\
& =\overline{\hat{{\rm p}}}_{i}-\hat{\boldsymbol{T}}\,\overline{y}_{i}\label{eq:SLAM_e_Final}
\end{align}
Accordingly, $\overset{\circ}{e}_{i}=[(\tilde{{\rm p}}_{i}-\tilde{P})^{\top},0]^{\top}$
where $\tilde{{\rm p}}_{i}=\hat{{\rm p}}_{i}-\tilde{R}{\rm p}_{i}$
represents the $i$th error in feature estimation, and $\tilde{P}=\hat{P}-\tilde{R}P$
as expressed in \eqref{eq:SLAM_T_error}. To achieve adaptive estimation,
let $\hat{b}_{U}=[\hat{b}_{\Omega}^{\top},\hat{b}_{V}^{\top}]^{\top}$
be the estimate of the unknown bias $b_{U}$ and let the error between
them be
\begin{equation}
\tilde{b}_{U}=b_{U}-\hat{b}_{U}\in\mathbb{R}^{6}\label{eq:SLAM_b_error}
\end{equation}
with $\tilde{b}_{U}=[\tilde{b}_{\Omega}^{\top},\tilde{b}_{V}^{\top}]^{\top}$.
As mentioned previously, the true SLAM kinematics in \eqref{eq:SLAM_True_dot}
are nonlinear modeled on Lie group of $\mathbb{SLAM}_{n}\left(3\right)=\mathbb{SE}\left(3\right)\times\overline{\mathcal{M}}^{n}$
such that $X=(\boldsymbol{T},\overline{{\rm p}})\in\mathbb{SLAM}_{n}\left(3\right)$.
Also, the tangent space of $X$ is $\mathfrak{slam}_{n}\left(3\right)=\mathfrak{se}\left(3\right)\times\overset{\circ}{\mathcal{M}}^{n}$
with $\mathcal{Y}=([U]_{\wedge},\overset{\circ}{{\rm v}})\in\mathfrak{slam}_{n}\left(3\right)$.
Therefore, the proposed observer design has to 1) consider the nonlinear
nature of the true SLAM problem and 2) be modeled on Lie group of
$\mathbb{SLAM}_{n}\left(3\right)$. Therefore, the observer proposed
in the next section is defined by $\hat{X}=(\hat{\boldsymbol{T}},\overline{\hat{{\rm p}}})\in\mathbb{SLAM}_{n}\left(3\right)$
mimics the structure of the true SLAM problem with its tangent space
being $\hat{\mathcal{Y}}=([\hat{U}]_{\wedge},\overset{\circ}{\hat{{\rm v}}})\in\mathfrak{slam}_{n}\left(3\right)$.

\section{Nonlinear Observer Design with Guaranteed Performance \label{sec:SLAM_Filter}}

This section reformulates the SLAM kinematics such that the error
function is guided by prescribed measures of transient and steady-state
performance. Next, nonlinear observer design characterized by systematic
convergence and reliant on available measurements is proposed.

\subsection{Guaranteed Performance\label{subsec:Systematic_Convergence}}

The key objective of this subsection is to force $e_{i}=[e_{i,1},e_{i,2},e_{i,3}]^{\top}$
described in \eqref{eq:SLAM_e_Final} to obey dynamically reducing
transient boundaries and settle down within the narrow bounds adjusted
by the user. A positive and time decreasing prescribed performance
function (PPF) $\xi_{i,k}\left(t\right)$ with the map of $\xi_{i,k}:\mathbb{R}_{+}\to\mathbb{R}_{+}$
\cite{bechlioulis2008robust,hashim2019SO3Wiley} is employed to guide
$e_{i,k}$ to initiate within a given large set $\xi_{i}^{0}=\left[\xi_{i,1}^{0},\xi_{i,2}^{0},\xi_{i,3}^{0}\right]^{\top}\in\mathbb{R}^{3}$
and decay exponentially in accordance with a known convergence factor
$\ell_{i}=\left[\ell_{i,1},\ell_{i,2},\ell_{i,3}\right]^{\top}\in\mathbb{R}^{3}$
causing $e_{i,k}$ to stay within a given small set $\xi_{i}^{\infty}=\left[\xi_{i,1}^{\infty},\xi_{i,2}^{\infty},\xi_{i,3}^{\infty}\right]^{\top}\in\mathbb{R}^{3}$
such that
\begin{equation}
\xi_{i,k}\left(t\right)=\left(\xi_{i,k}^{0}-\xi_{i,k}^{\infty}\right)\exp\left(-\ell_{i,k}t\right)+\xi_{i,k}^{\infty}\label{eq:SLAM_Presc}
\end{equation}
for all $i=1,2,\ldots,n$ and $k=1,2,3$. The objective is $e_{i,k}=e_{i,k}\left(t\right)$
follows predefined convergence properties of $\xi_{i,k}=\xi_{i,k}\left(t\right)$
given that one of the following expressions is met:
\begin{align}
-\delta_{i,k}\xi_{i,k}<e_{i,k}<\xi_{i,k}, & \text{ if }e_{i,k}\left(0\right)\geq0\label{eq:SLAM_ePos}\\
-\xi_{i,k}<e_{i,k}<\delta_{i,k}\xi_{i,k}, & \text{ if }e_{i,k}\left(0\right)<0\label{eq:SLAM_eNeg}
\end{align}
where $\delta_{i,k}\in\left[0,1\right]$. Fig. \ref{fig:SO3PPF_2}
provides an ample demonstration of the desired systematic convergence. 

\begin{figure}[h!]
	\centering{}\includegraphics[scale=0.32]{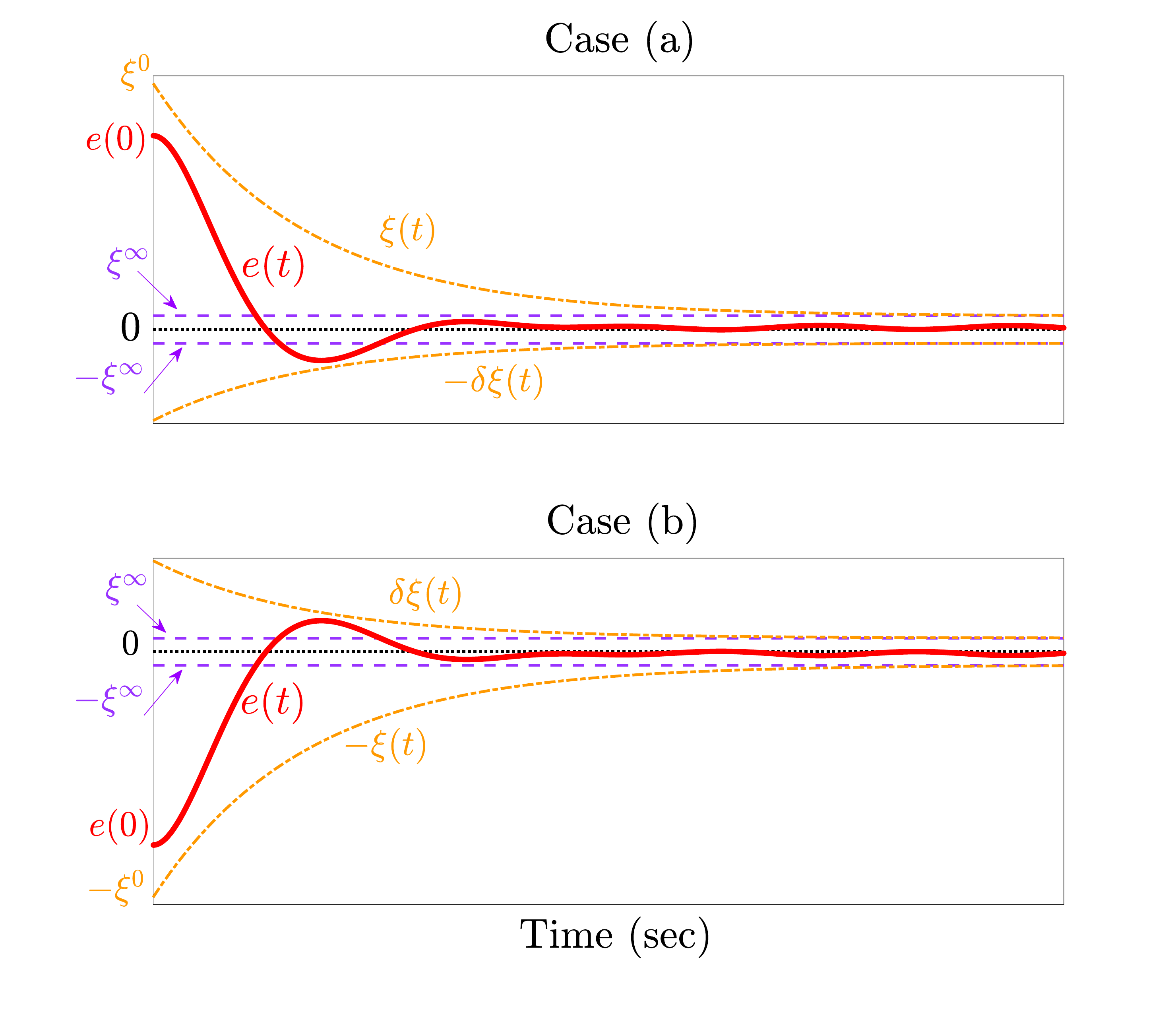} \caption{$e_{i,k}$ based on systematic convergence (a) Eq. \eqref{eq:SLAM_ePos};
		(b) Eq. \eqref{eq:SLAM_eNeg}.}
	\label{fig:SO3PPF_2}
\end{figure}

\begin{rem}
	\label{SE3PPF_rem3}\cite{hashim2019SO3Wiley} For known $e_{i,k}\left(0\right)$
	and granted that either provision \eqref{eq:SLAM_ePos} and \eqref{eq:SLAM_eNeg}
	is fulfilled, the maximum undershoot/overshoot is guaranteed to adhere
	$\pm\delta\xi_{i,k}$ and the steady-state error to follows $\pm\xi_{i,k}^{\infty}$
	in accordance with Fig. \ref{fig:SO3PPF_2}. 
\end{rem}
Define the error $e_{i,k}$ as
\begin{equation}
e_{i,k}=\xi_{i,k}\mathcal{F}(E_{i,k})\label{eq:SLAM_e_Trans}
\end{equation}
where $\xi_{i,k}\in\mathbb{R}$ is as in \eqref{eq:SLAM_Presc}, $E_{i,k}\in\mathbb{R}$
denotes unconstrained or transformed error, and $\mathcal{F}(E_{i,k})$
describes a smooth function that adheres to Assumption \ref{Assum:SE3PPF_1}:

\begin{assum}\label{Assum:SE3PPF_1} $\mathcal{F}(E_{i,k})$ is a
	smooth function with the following characteristics \cite{bechlioulis2008robust}:
	
	\begin{enumerate}
		\item[1)] Strictly increasing, 
		\item[2)] Constrained by
		\[
		\begin{cases}
		-\delta_{i,k}<\mathcal{F}(E_{i,k})<\delta_{i,k}, & \text{ if }e_{i,k}\left(0\right)\geq0\\
		-\bar{\delta}_{i,k}<\mathcal{F}(E_{i,k})<\underline{\delta}_{i,k}, & \text{ if }e_{i,k}\left(0\right)<0
		\end{cases}
		\]
		with $\bar{\delta}_{i,k},\underline{\delta}_{i,k}>0$ and $\underline{\delta}_{i,k}\leq\bar{\delta}_{i,k}$ 
		\item[3)] $\underset{E_{i,k}\rightarrow-\infty}{\lim}\mathcal{F}(E_{i,k})=-\underline{\delta}_{i,k},$
		and $\underset{E_{i,k}\rightarrow+\infty}{\lim}\mathcal{F}(E_{i,k})=\bar{\delta}_{i,k}$
		$\text{if }e_{i,k}\left(0\right)\geq0$ or $\underset{E_{i,k}\rightarrow-\infty}{\lim}\mathcal{F}(E_{i,k})=-\bar{\delta}_{i,k}$
		and $\underset{E_{i,k}\rightarrow+\infty}{\lim}\mathcal{F}(E_{i,k})=\underline{\delta}_{i,k}$
		$\text{if }e_{i,k}\left(0\right)<0$.
	\end{enumerate}
\end{assum}Define $\mathcal{F}\left(E_{i,k}\right)$ as below
\begin{equation}
\mathcal{F}\left(E_{i,k}\right)=\frac{\bar{\delta}_{i,k}\exp(E_{i,k})-\underline{\delta}_{i,k}\exp(-E_{i,k})}{\exp(E_{i,k})+\exp(-E_{i,k})}\label{eq:SLAM_Smooth}
\end{equation}
with $\bar{\delta}_{i,k}\geq\underline{\delta}_{i,k}$ if $e_{i,k}\left(0\right)\geq0$
and $\underline{\delta}_{i,k}\geq\bar{\delta}_{i,k}$ if $e_{i,k}\left(0\right)<0$.

The inverse transformation of \eqref{eq:SLAM_Smooth} gives $E_{i,k}=E_{i,k}(e_{i,k},\xi_{i,k})$
\begin{equation}
\begin{aligned}E_{i,k}= & \mathcal{F}^{-1}(e_{i,k}/\xi_{i,k})=\frac{1}{2}\text{ln}\frac{\underline{\delta}_{i,k}+e_{i,k}/\xi_{i,k}}{\bar{\delta}_{i,k}-e_{i,k}/\xi_{i,k}}\end{aligned}
\label{eq:SLAM_trans2}
\end{equation}
where $\bar{\delta}_{i,k}\geq\underline{\delta}_{i,k}$ if $e_{i,k}\left(0\right)\geq0$
and $\underline{\delta}_{i,k}\geq\bar{\delta}_{i,k}$ if $e_{i,k}\left(0\right)<0$.
Let
\begin{align}
\eta_{i,k} & =\frac{1}{2\xi_{i,k}}\frac{\partial\mathcal{F}^{-1}(e_{i,k}/\xi_{i,k})}{\partial(e_{i,k}/\xi_{i,k})}\nonumber \\
& =\frac{1}{2\xi_{i,k}}\left(\frac{1}{\underline{\delta}_{i,k}+e_{i,k}/\xi_{i,k}}+\frac{1}{\bar{\delta}_{i,k}-e_{i,k}/\xi_{i,k}}\right)\label{eq:SLAM_Aux1}
\end{align}
for all $i=1,2,\ldots,n$ and $k=1,2,3$. In the light of \eqref{eq:SLAM_Aux1},
define
\begin{align}
\begin{cases}
\mu_{i} & ={\rm diag}\left(\frac{\dot{\xi}_{i,1}}{\xi_{i,1}},\frac{\dot{\xi}_{i,2}}{\xi_{i,2}},\frac{\dot{\xi}_{i,3}}{\xi_{i,3}}\right)\\
\Lambda_{i} & ={\rm diag}\left(\eta_{i,1},\eta_{i,2},\eta_{i,3}\right)
\end{cases} & ,\hspace{1em}\forall i=1,2,\ldots,n\label{eq:SLAM_Aux2}
\end{align}
To this end, the transformed error dynamics of $E_{i}=\left[E_{i,1},E_{i,2},E_{i,3}\right]^{\top}\in\mathbb{R}^{3}$
become equivalent to
\begin{equation}
\dot{E}_{i}=\Lambda_{i}\left(\dot{e}_{i}-\mu_{i}e_{i}\right),\hspace{1em}\forall i=1,2,\ldots,n\label{eq:SLAM_Trans_dot}
\end{equation}
Note that $\mu_{i}$ is a vanishing element where $\mu_{i}\rightarrow0$
as $t\rightarrow\infty$.

\subsection{Nonlinear Observer Design\label{subsec:Det_without_IMU}}

Consider the following nonlinear observer

\begin{align}
\dot{\hat{\boldsymbol{T}}} & =\hat{\boldsymbol{T}}\left[U_{m}-\hat{b}_{U}-W_{U}\right]_{\wedge}\label{eq:SLAM_T_est_dot_f2}\\
\dot{{\rm \hat{p}}}_{i} & =-k_{p}\left(\Lambda_{i}+\Lambda_{i}^{-1}\right)E_{i},\hspace{1em}i=1,2,\ldots,n\label{eq:SLAM_p_est_dot_f2}\\
\dot{\hat{b}}_{U} & =-\sum_{i=1}^{n}\frac{\Gamma}{\alpha_{i}}\overline{{\rm Ad}}_{\hat{\boldsymbol{T}}}^{\top}\left[\begin{array}{c}
\left[\hat{R}y_{i}+\hat{P}\right]_{\times}\\
\mathbf{I}_{3}
\end{array}\right]\Lambda_{i}E_{i}\label{eq:SLAM_b_est_dot_f2}\\
W_{U} & =-\sum_{i=1}^{n}k_{w}\overline{{\rm Ad}}_{\hat{\boldsymbol{T}}^{-1}}\left[\begin{array}{c}
\left[\hat{R}y_{i}+\hat{P}\right]_{\times}\\
\mathbf{I}_{3}
\end{array}\right]\Lambda_{i}E_{i}\label{eq:SLAM_W_f2}
\end{align}
where $\Lambda_{i}$ and ${\rm \mu_{i}}$ are defined in \eqref{eq:SLAM_Aux2},
$k_{p}$, $k_{w}$, $\Gamma$, and $\alpha_{i}$ are positive constants,
$W_{U}=\left[W_{\Omega}^{\top},W_{V}^{\top}\right]^{\top}\in\mathbb{R}^{6}$
denotes a correction factor, $\hat{b}_{U}=\left[\hat{b}_{\Omega}^{\top},\hat{b}_{V}^{\top}\right]^{\top}\in\mathbb{R}^{6}$
is the estimate of $b_{U}$ for all $W_{\Omega},W_{V},\hat{b}_{\Omega},\hat{b}_{V}\in\mathbb{R}^{3}$.
The unit-quaternion representation of the proposed observer is presented
in \nameref{sec:SO3_PPF_STCH_AppendixA}.
\begin{thm}
	\label{thm:PPF}Consider the SLAM dynamics $\dot{X}=\left(\dot{\boldsymbol{T}},\dot{\overline{{\rm p}}}\right)$
	in \eqref{eq:SLAM_True_dot} combined with velocity measurements ($U_{m}=U+b_{U}$)
	and output measurements ($\overline{y}_{i}=\boldsymbol{T}^{-1}\overline{{\rm p}}_{i}$)
	for all $i=1,2,\ldots,n$. Suppose that Assumption \ref{Assumption:Feature}
	holds and the observer design is as in \eqref{eq:SLAM_T_est_dot_f2},
	\eqref{eq:SLAM_p_est_dot_f2}, \eqref{eq:SLAM_b_est_dot_f2}, and
	\eqref{eq:SLAM_W_f2}. Select the design parameters $k_{p}$, $k_{w}$,
	$\Gamma$, and $\alpha_{i}$ as positive constants and $\bar{\delta}_{i,k}=\underline{\delta}_{i,k}$
	$\forall i=1,2,\ldots,n$ and $k=1,2,3$. Define the following set
	\begin{align}
	\mathcal{S}= & \{(E_{1},E_{1},\ldots,E_{n})\in\mathbb{R}^{3}\times\mathbb{R}^{3}\times\cdots\times\mathbb{R}^{3}\nonumber \\
	& \hspace{7.5em}|E_{i}=\underline{\mathbf{0}}_{3}\forall i=1,2,\ldots,n\}\label{eq:SLAM_Set2}
	\end{align}
	Then, for $E_{i}\left(0\right)\ensuremath{\in\mathcal{L}_{\infty}}$,
	(1) the error $(E_{1},E_{2},\ldots,E_{n})$ exponentially approaches
	$\mathcal{S}$, (2) the error $(e_{1},e_{2},\ldots,e_{n})$ asymptotically
	approaches $(\underline{\mathbf{0}}_{3},\underline{\mathbf{0}}_{3},\ldots,\underline{\mathbf{0}}_{3})$,
	(3) $\tilde{b}_{U}$ asymptotically converges to the origin, and (4)
	there exists a constant matrix $R_{c}\in\mathbb{SO}\left(3\right)$
	and a constant vector $P_{c}\in\mathbb{R}^{3}$ with $\lim_{t\rightarrow\infty}\tilde{R}=R_{c}$
	and $\lim_{t\rightarrow\infty}\tilde{P}=P_{c}$.
\end{thm}
\begin{proof} Consider the pose error described in \eqref{eq:SLAM_T_error}.
	The pose error dynamics are
	\begin{align}
	\dot{\tilde{\boldsymbol{T}}} & =\dot{\hat{\boldsymbol{T}}}\boldsymbol{T}^{-1}+\hat{\boldsymbol{T}}\dot{\boldsymbol{T}}^{-1}\nonumber \\
	& =\hat{\boldsymbol{T}}\left[U_{m}-\hat{b}_{U}-W_{U}\right]_{\wedge}\boldsymbol{T}^{-1}-\hat{\boldsymbol{T}}\boldsymbol{T}^{-1}\boldsymbol{\dot{T}}\boldsymbol{T}^{-1}\nonumber \\
	& =\hat{\boldsymbol{T}}\left[U+\tilde{b}_{U}-W_{U}\right]_{\wedge}\boldsymbol{T}^{-1}-\hat{\boldsymbol{T}}\left[U\right]_{\wedge}\boldsymbol{T}^{-1}\nonumber \\
	& ={\rm Ad}_{\hat{\boldsymbol{T}}}\left(\left[\tilde{b}_{U}-W_{U}\right]_{\wedge}\right)\tilde{\boldsymbol{T}}\label{eq:SLAM_T_error_dot}
	\end{align}
	Note that $\boldsymbol{\dot{T}}^{-1}=-\boldsymbol{T}^{-1}\boldsymbol{\dot{T}}\boldsymbol{T}^{-1}$.
	As such, the dynamics of $\overset{\circ}{e}_{i}$ in \eqref{eq:SLAM_e}
	are
	\begin{align}
	\overset{\circ}{\dot{e}}_{i} & =\overset{\circ}{\dot{\hat{{\rm p}}}}_{i}-\dot{\tilde{\boldsymbol{T}}}\,\overline{{\rm p}}_{i}-\tilde{\boldsymbol{T}}\,\dot{\overline{{\rm p}}}_{i}\nonumber \\
	& =\overset{\circ}{\dot{\hat{{\rm p}}}}_{i}-{\rm Ad}_{\hat{\boldsymbol{T}}}\left(\left[\tilde{b}_{U}-W_{U}\right]_{\wedge}\right)\tilde{\boldsymbol{T}}\,\overline{{\rm p}}_{i}\label{eq:SLAM_e_dot}
	\end{align}
	In the light of expressions in \eqref{eq:SLAM_Adjoint} and \eqref{eq:SLAM_Adjoint_MAP},
	one finds ${\rm Ad}_{\hat{\boldsymbol{T}}}\left(\left[\tilde{b}_{U}-W_{U}\right]_{\wedge}\right)=\left[\overline{{\rm Ad}}_{\hat{\boldsymbol{T}}}(\tilde{b}_{U}-W_{U})\right]_{\wedge}$
	such that
	\begin{align}
	{\rm Ad}_{\hat{\boldsymbol{T}}}([\tilde{b}_{U}]_{\wedge})\tilde{\boldsymbol{T}}\,\overline{{\rm p}}_{i} & =\left[\begin{array}{c}
	\left[\hat{R}y_{i}+\hat{P}\right]_{\times}\\
	\mathbf{I}_{3}
	\end{array}\right]^{\top}\overline{{\rm Ad}}_{\hat{\boldsymbol{T}}}\tilde{b}_{U}\label{eq:SLAM_expression1}
	\end{align}
	Accordingly, the expression in \eqref{eq:SLAM_e_dot} becomes
	\begin{align}
	\dot{e}_{i} & =\dot{\hat{{\rm p}}}_{i}-\left[\begin{array}{c}
	\left[\hat{R}y_{i}+\hat{P}\right]_{\times}\\
	\mathbf{I}_{3}
	\end{array}\right]^{\top}\overline{{\rm Ad}}_{\hat{\boldsymbol{T}}}\left(\tilde{b}_{U}-W_{U}\right)\label{eq:SLAM_e_dot_Final}
	\end{align}
	Hence, the transformed error dynamics $\dot{E}_{i}=\Lambda_{i}\left(\dot{e}_{i}-\mu_{i}e_{i}\right)$
	can be obtained by \eqref{eq:SLAM_Trans_dot}, given \eqref{eq:SLAM_e}
	and \eqref{eq:SLAM_e_dot_Final} for all $i=1,2,\ldots,n$. Define
	the following candidate Lyapunov function $\mathcal{L}=\mathcal{L}(E_{1},E_{2},\ldots,E_{n},\tilde{b}_{U})$
	\begin{equation}
	\mathcal{L}=\sum_{i=1}^{n}\frac{1}{2\alpha_{i}}\left\Vert E_{i}\right\Vert ^{2}+\frac{1}{2}\tilde{b}_{U}^{\top}\Gamma^{-1}\tilde{b}_{U}\label{eq:SLAM_Lyap1}
	\end{equation}
	From \eqref{eq:SLAM_Trans_dot} and \eqref{eq:SLAM_Lyap1}, and differentiating
	$\mathcal{L}$ one obtains
	\begin{align}
	\dot{\mathcal{L}}= & \sum_{i=1}^{n}\frac{1}{\alpha_{i}}E_{i}^{\top}\dot{E}_{i}-\tilde{b}_{U}^{\top}\Gamma^{-1}\dot{\hat{b}}_{U}\nonumber \\
	= & -\sum_{i=1}^{n}\frac{1}{\alpha_{i}}E_{i}^{\top}\Lambda_{i}\left[\begin{array}{c}
	\left[\hat{R}y_{i}+\hat{P}\right]_{\times}\\
	\mathbf{I}_{3}
	\end{array}\right]^{\top}\overline{{\rm Ad}}_{\hat{\boldsymbol{T}}}\left(\tilde{b}_{U}-W_{U}\right)\nonumber \\
	& +\sum_{i=1}^{n}\frac{1}{\alpha_{i}}E_{i}^{\top}\Lambda_{i}\left(\dot{\hat{{\rm p}}}_{i}-\mu_{i}e_{i}\right)-\tilde{b}_{U}^{\top}\Gamma^{-1}\dot{\hat{b}}_{U}\label{eq:SLAM_Lyap2_dot2}
	\end{align}
	By \eqref{eq:SLAM_e_Trans} and \eqref{eq:SLAM_trans2} $|e_{i,k}|\leq\mu_{i,k}\bar{\delta}_{i,k}\xi_{i,k}|E_{i,k}|$,
	moreover, $\mu_{i,k}$ is a vanishing component. Consider $\bar{k}_{\delta}=\max\{\bar{\delta}_{1,1},\bar{\delta}_{1,2},\ldots,\bar{\delta}_{n,3}\}$,
	$\bar{k}_{\xi}=\max\{\xi_{1,1}^{0},\xi_{1,2}^{0},\ldots,\xi_{n,1}^{0}\}$,
	and the negative vanishing component $\bar{\mu}=\min\{\mu_{1,1},\mu_{1,2},\ldots,\mu_{n,3}\}\leq0$.
	Substituting $W_{U}$, $\dot{\hat{b}}_{U}$ and $\dot{\hat{{\rm p}}}_{i}$
	with their definitions in \eqref{eq:SLAM_W_f2}, \eqref{eq:SLAM_b_est_dot_f2},
	and \eqref{eq:SLAM_p_est_dot_f2}, respectively, one obtains
	\begin{align}
	\dot{\mathcal{L}}\leq & -c_{p}\sum_{i=1}^{n}\frac{1}{\alpha_{i}}\left\Vert E_{i}\right\Vert ^{2}-k_{w}\left\Vert \sum_{i=1}^{n}\frac{1}{\alpha_{i}}\Lambda_{i}E_{i}\right\Vert ^{2}\nonumber \\
	& -k_{w}\left\Vert \sum_{i=1}^{n}\frac{1}{\alpha_{i}}\left[\hat{R}y_{i}+\hat{{\rm p}}_{i}\right]_{\times}\Lambda_{i}E_{i}\right\Vert ^{2}\label{eq:SLAM_Lyap2_final}
	\end{align}
	where $c_{p}=k_{p}-\bar{k}_{\delta}\bar{k}_{\xi}|\bar{\mu}|$ such
	that $k_{p}>\bar{k}_{\delta}\bar{k}_{\xi}|\bar{\mu}|$. Based on \eqref{eq:SLAM_Lyap2_final},
	$\dot{\mathcal{L}}$ is negative definite such that $\mathcal{L}\rightarrow0$
	which in turn implies that $(E_{1},E_{2},\ldots,E_{n})$ converges
	asymptotically to $\mathcal{S}$ defined in \eqref{eq:SLAM_Set2}
	for all $E_{i}\left(0\right)\in\mathcal{L}_{\infty}$ proving item
	(1) in Theorem \ref{thm:PPF}. It becomes apparent that $\mathcal{L}\in\mathcal{L}_{\infty}$
	and that a finite $\lim_{t\rightarrow\infty}\mathcal{L}$ exists.
	Given that $\bar{\delta}_{i,k}=\underline{\delta}_{i,k}$, and in
	the light of \eqref{eq:SLAM_e_Trans} and \eqref{eq:SLAM_trans2},
	it is given that
	\begin{align*}
	e_{i,k} & =\bar{\delta}_{i,k}\xi_{i,k}\frac{\exp(E_{i,k})-\exp(-E_{i,k})}{\exp(E_{i,k})+\exp(-E_{i,k})},\hspace{1em}\bar{\delta}_{i,k}=\underline{\delta}_{i,k}
	\end{align*}
	This implies that $E_{i,k}\neq0$ for $e_{i,k}\neq0$ and $E_{i,k}=0$
	only at $e_{i,k}=0$ proving item (2) in Theorem \ref{thm:PPF}. The
	fact that $E_{i,k}$ and $e_{i,k}$ converge to zero indicates that
	$\ddot{E}_{i}$ and $\ddot{e}_{i}$ remain bounded, and thereby $\dot{E}_{i}$
	and $\dot{e}_{i}$ are uniformly continuous. Based on Barbalat Lemma,
	$\dot{E}_{i}\rightarrow0$ and $\dot{e}_{i}\rightarrow0$ as $t\rightarrow\infty$.
	According to the definition of $\tilde{b}_{U}$ in \eqref{eq:SLAM_b_error}
	along with \eqref{eq:SLAM_b_est_dot_f2}, $\dot{\tilde{b}}_{U}=-\dot{\hat{b}}_{U}$,
	and as a result $\dot{\tilde{b}}_{U}\rightarrow0$ as $E_{i}\rightarrow0$.%
	{} Also, from \eqref{eq:SLAM_W_f2}, $W_{U}\rightarrow0$ as $E_{i}\rightarrow0$.
	Additionally from \eqref{eq:SLAM_p_est_dot_f2}, $\dot{{\rm \hat{p}}}_{i}\rightarrow0$
	as $E_{i}\rightarrow0$ and $e_{i}\rightarrow0$. Consequently, $\lim_{t\rightarrow\infty}\overset{\circ}{\dot{e}}_{i}=\lim_{t\rightarrow\infty}-{\rm Ad}_{\hat{\boldsymbol{T}}}\left(\left[\tilde{b}_{U}\right]_{\wedge}\right)\tilde{\boldsymbol{T}}\,\overline{{\rm p}}_{i}=0$
	that is $\lim_{t\rightarrow\infty}\overset{\circ}{\dot{e}}_{i}=\lim_{t\rightarrow\infty}-\hat{\boldsymbol{T}}\left[\tilde{b}_{U}\right]_{\wedge}\overline{y}_{i}=0$.
	It follows that $\lim_{t\rightarrow\infty}\left[\tilde{b}_{U}\right]_{\wedge}\overline{y}_{i}=\lim_{t\rightarrow\infty}\left[-\left[y_{i}\right]_{\times},\mathbf{I}_{3}\right]\tilde{b}_{U}=0$
	for all $i=1,2,\ldots,n$. Let 
	\[
	M=\left[\begin{array}{cc}
	-\left[y_{1}\right]_{\times} & \mathbf{I}_{3}\\
	\vdots & \vdots\\
	-\left[y_{n}\right]_{\times} & \mathbf{I}_{3}
	\end{array}\right]\in\mathbb{R}^{3n\times6},\hspace{1em}n\geq3
	\]
	As specified in Assumption \ref{Assumption:Feature}, number of features
	is greater than or equal to 3. Thus $M$ has full column rank and
	$\lim_{t\rightarrow\infty}M\tilde{b}_{U}=0$ signifying that $\lim_{t\rightarrow\infty}\tilde{b}_{U}=0$
	showing item (3) in Theorem \ref{thm:PPF}. Accordingly, from \eqref{eq:SLAM_Lyap2_final},
	$\ddot{\mathcal{L}}$ is bounded. In the light of Barbalat Lemma,
	$\dot{\mathcal{L}}$ is uniformly continuous. Since both $\tilde{b}_{U}\rightarrow0$
	and $W_{U}\rightarrow0$ as $t\rightarrow\infty$, $\dot{\tilde{\boldsymbol{T}}}\rightarrow0$
	and in turn $\tilde{\boldsymbol{T}}\rightarrow\boldsymbol{T}_{c}(R_{c},P_{c})$
	where $\boldsymbol{T}_{c}(R_{c},P_{c})\in\mathbb{SE}\left(3\right)$
	denotes a constant matrix with $R_{c}\in\mathbb{SO}\left(3\right)$
	and $P_{c}\in\mathbb{R}^{3}$. Thus, one concludes that $\lim_{t\rightarrow\infty}\tilde{R}=R_{c}$
	and $\lim_{t\rightarrow\infty}\tilde{P}=P_{c}$ completing the proof.\end{proof}

\section{Simulation results \label{sec:SE3_Simulations}}

This section explores the performance of the nonlinear observer for
SLAM on the Lie group $\mathbb{SLAM}_{n}\left(3\right)$ with systematic
convergence. Consider the angular velocity to be $\Omega=[0,0,0.2]^{\top}({\rm rad/sec})$
and the translational velocity to be $V=[1.8,0,0]^{\top}({\rm m/sec})$.
Let the true initial attitude and position of the vehicle be $R\left(0\right)=\mathbf{I}_{3}$
and $P\left(0\right)=[0,0,3]^{\top}$, respectively. Additionally,
consider four features fixed with respect to the inertial-frame at
the following locations: ${\rm p}_{1}=[8,8,0]^{\top}$, ${\rm p}_{2}=[-8,8,0]^{\top}$,
${\rm p}_{3}=[8,-8,0]^{\top}$, and ${\rm p}_{4}=[-8,-8,0]^{\top}$.
In practice, $b_{U}$ and $n_{U}$ are non-zero. Hence, let the group
velocity vector bias be $b_{U}=\left[b_{\Omega}^{\top},b_{V}^{\top}\right]^{\top}$
with $b_{\Omega}=[0.09,0.1,-0.1]^{\top}({\rm rad/sec})$ and $b_{V}=[0.2,0.2,-0.2]^{\top}({\rm m/sec})$,
and noise $n_{U}$ of zero mean and standard deviation of 0.2. Let
the initial estimate of attitude and position be $\hat{R}\left(0\right)=\mathbf{I}_{3}$
and $\hat{P}\left(0\right)=[0,0,0]^{\top}$, respectively, and let
the initial estimates of the four features be $\hat{{\rm p}}_{1}\left(0\right)=\hat{{\rm p}}_{2}\left(0\right)=\hat{{\rm p}}_{3}\left(0\right)=\hat{{\rm p}}_{4}\left(0\right)=[0,0,0]^{\top}$.
Design parameters and initial bias estimate are chosen as follows:
$\alpha_{i}=0.05$, $\Gamma=10\mathbf{I}_{6}$, $k_{w}=3$, $k_{p}=3$,
$\ell_{i,k}=1$, $\xi_{i,k}^{\infty}=0.1$, $\xi_{i,k}^{0}=\bar{\delta}_{i,k}=\underline{\delta}_{i,k}=1.2e_{i,k}\left(0\right)+1.8$,
and $\hat{b}_{U}\left(0\right)=\underline{\mathbf{0}}_{6}$ for all
$i=1,2,3,4$ and $k=1,2,3$. 

Fig. \ref{fig:SLAM_3d} depicts the true and estimated trajectories
of the vehicle and the position of the features. The true vehicle
trajectory is plotted as a solid black line with a black circle marking
the final destination. The true feature positions are marked as black
circles at ${\rm p}_{1}$, ${\rm p}_{2}$, ${\rm p}_{3}$ and ${\rm p}_{4}$.
Blue and red are used for the observer output.  The estimated trajectory
of the vehicle is represented by a blue dashed line which tracks the
travel path from the origin (0,0,0) to its final destination marked
with a blue star\textcolor{blue}{{} $\star$}. The feature position
estimates, indicated by the red dashed lines, initiate at the origin
(0,0,0) and then gradually diverge to the true feature locations marked
with red stars \textcolor{red}{$\star$}. Both vehicle trajectory
and feature positions commence at the origin with large initialization
error and converge successfully to the true trajectory and locations,
respectively. As such, Fig. \ref{fig:SLAM_3d} reveals impressive
tracking capabilities of the proposed observer.

\begin{figure}[h]
	\centering{}\includegraphics[scale=0.29]{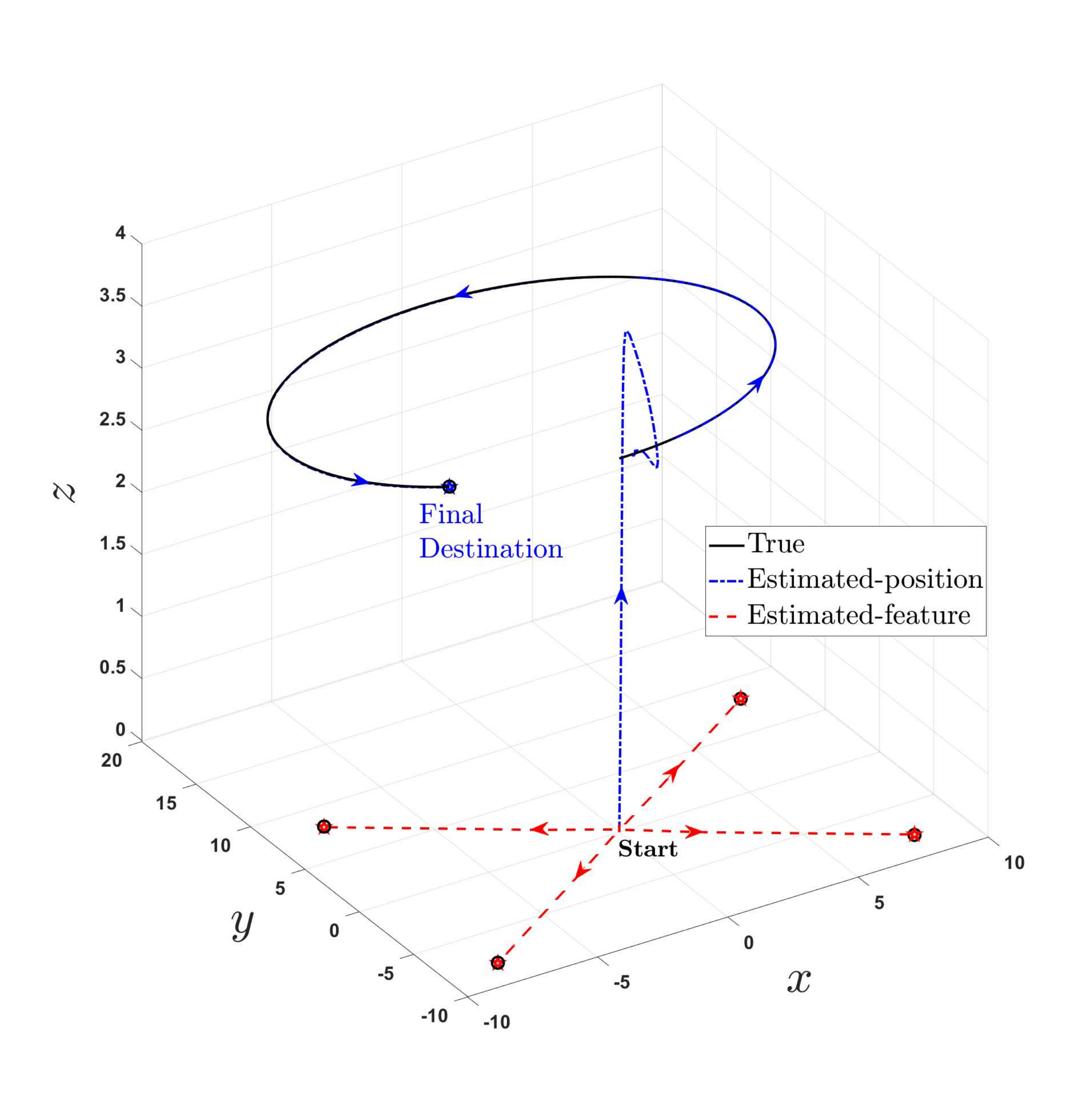}\caption{Output trajectories of the observer vs true vehicle's and features
		position.}
	\label{fig:SLAM_3d}
\end{figure}

Fig. \ref{fig:SLAM_error} illustrates the error trajectories of $e_{i}=[e_{i1},e_{i2},e_{i3}]^{\top}$
for $i=1,2,3,4$ plotted in red, blue, and magenta with respect to
the dynamically reducing boundaries of PPF plotted in black. As shown
in Fig. \ref{fig:SLAM_error} large initial error does not surpass
the boundaries of the predefined large set and reduces following the
dynamically reducing boundaries to a predefined small set. Therefore,
the simulation results align with the theoretical results and demonstrate
outstanding estimation capability of the proposed observer.

\begin{figure}[h]
	\centering{}\includegraphics[scale=0.3]{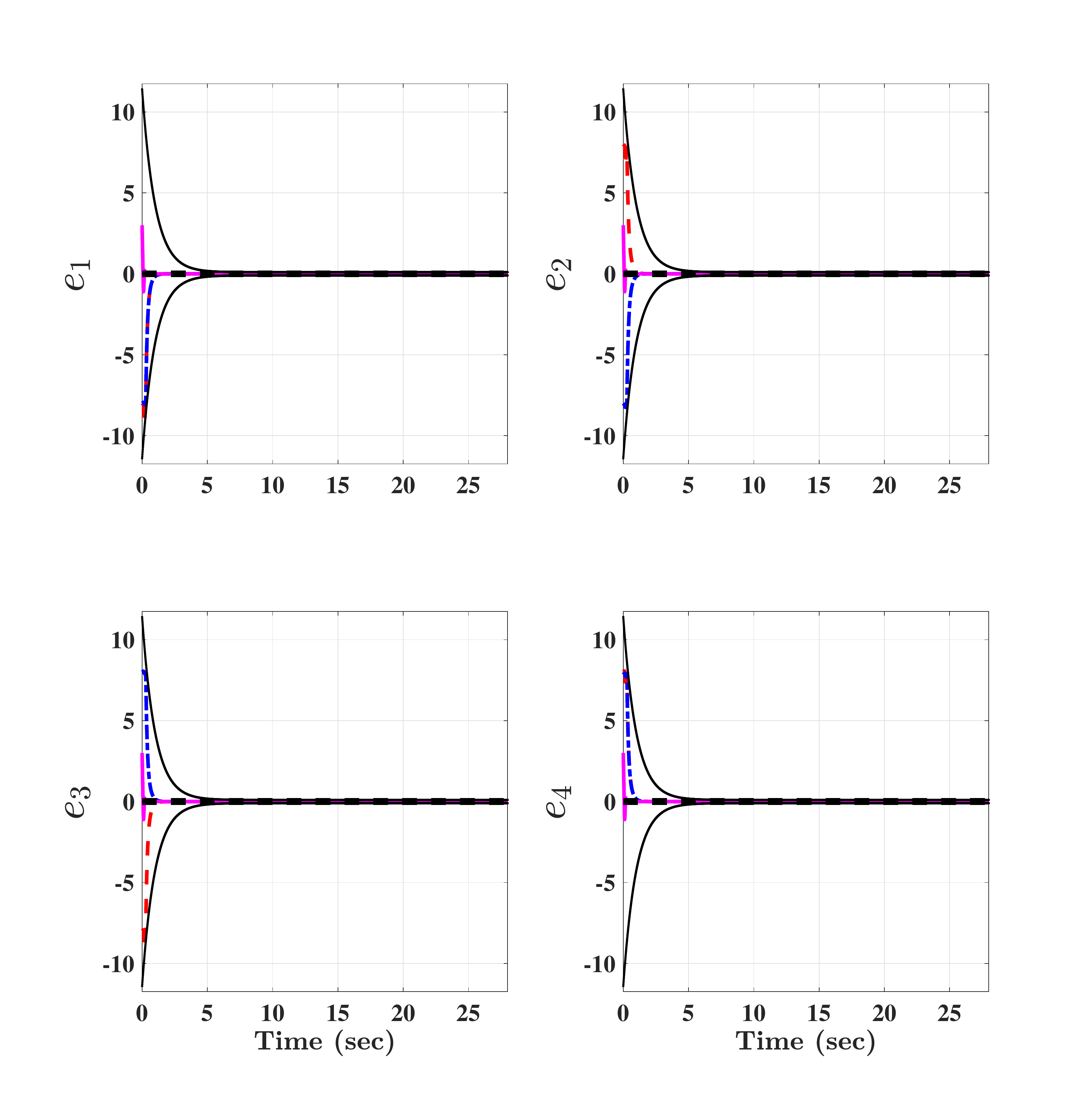}\caption{Error trajectories of $e_{i}=[e_{i1},e_{i2},e_{i3}]^{\top}$for $i=1,2,3,4$
		with respect to dynamically reducing boundaries of PPF.}
	\label{fig:SLAM_error}
\end{figure}

\section{Conclusion \label{sec:SE3_Conclusion}}

This paper presented a novel nonlinear observer for Simultaneous Localization
and Mapping (SLAM) problem on the Lie group of $\mathbb{SLAM}_{n}\left(3\right)$.
The observer has been developed such that the error function is guaranteed
to follow predefined measures of transient and steady-state performance.
Moreover, it is able to compensate for unknown bias attached to angular
and translational velocities. As has been demonstrated in the Simulation
Section, the proposed observer has the ability to produce reasonable
results localizing the unknown pose of the vehicle and concurrently
mapping the unknown environment with respect to available measurements
of angular velocity, translational velocity, and features obtained
in the body-frame.

\section*{Acknowledgment}

The authors would like to thank \textbf{Maria Shaposhnikova} for proofreading
the article.

\bibliographystyle{IEEEtran}
\bibliography{bib_SLAM}

\section*{Appendix\label{sec:SO3_PPF_STCH_AppendixA} }
\begin{center}
	\textbf{\large{}{}{}{}{}{}{}{}{}{}{}{}Quaternion Representation}{\large{}{}{}
	} 
	\par\end{center}

\noindent Define $Q=[q_{0},q^{\top}]^{\top}\in\mathbb{S}^{3}$ as
a unit-quaternion with $q_{0}\in\mathbb{R}$ and $q\in\mathbb{R}^{3}$
such that $\mathbb{S}^{3}=\{\left.Q\in\mathbb{R}^{4}\right|||Q||=\sqrt{q_{0}^{2}+q^{\top}q}=1\}$.
$Q^{-1}=[\begin{array}{cc}
q_{0} & -q^{\top}\end{array}]^{\top}\in\mathbb{S}^{3}$ denotes the inverse of $Q$. Define $\odot$ as a quaternion product
where the quaternion multiplication of $Q_{1}=[\begin{array}{cc}
q_{01} & q_{1}^{\top}\end{array}]^{\top}\in\mathbb{S}^{3}$ and $Q_{2}=[\begin{array}{cc}
q_{02} & q_{2}^{\top}\end{array}]^{\top}\in\mathbb{S}^{3}$ is 
\[
Q_{1}\odot Q_{2}=\left[\begin{array}{c}
q_{01}q_{02}-q_{1}^{\top}q_{2}\\
q_{01}q_{2}+q_{02}q_{1}+[q_{1}]_{\times}q_{2}
\end{array}\right]
\]
The mapping from unit-quaternion ($\mathbb{S}^{3}$) to $\mathbb{SO}\left(3\right)$
is described by $\mathcal{R}_{Q}:\mathbb{S}^{3}\rightarrow\mathbb{SO}\left(3\right)$
\begin{align}
\mathcal{R}_{Q} & =(q_{0}^{2}-||q||^{2})\mathbf{I}_{3}+2qq^{\top}+2q_{0}\left[q\right]_{\times}\in\mathbb{SO}\left(3\right)\label{eq:NAV_Append_SO3}
\end{align}
The quaternion identity is described by $Q_{{\rm I}}=[\pm1,0,0,0]^{\top}$
with $\mathcal{R}_{Q_{{\rm I}}}=\mathbf{I}_{3}$. Visit \cite{hashim2019AtiitudeSurvey}
for more information. Define the estimate of $Q=[q_{0},q^{\top}]^{\top}\in\mathbb{S}^{3}$
as $\hat{Q}=[\hat{q}_{0},\hat{q}^{\top}]^{\top}\in\mathbb{S}^{3}$
with 
\[
\mathcal{R}_{\hat{Q}}=(\hat{q}_{0}^{2}-||\hat{q}||^{2})\mathbf{I}_{3}+2\hat{q}\hat{q}^{\top}+2\hat{q}_{0}\left[\hat{q}\right]_{\times}\in\mathbb{SO}\left(3\right)
\]
see the map in \eqref{eq:NAV_Append_SO3}. Define the map
\begin{align*}
\left[\begin{array}{c}
0\\
\mathbf{Y}(\hat{Q},x)
\end{array}\right] & =\hat{Q}\odot\left[\begin{array}{c}
0\\
x
\end{array}\right]\odot\hat{Q}^{-1}
\end{align*}
where $\mathbf{Y}(\hat{Q},y_{i})\in\mathbb{R}^{3}$, $x\in\mathbb{R}^{3}$
and $\hat{Q}\in\mathbb{S}^{3}$. The equivalent quaternion representation
and complete implementation steps of the observer in \eqref{eq:SLAM_T_est_dot_f2},
\eqref{eq:SLAM_p_est_dot_f2}, \eqref{eq:SLAM_b_est_dot_f2}, and
\eqref{eq:SLAM_W_f2} is:
\[
\begin{cases}
\overset{\circ}{e}_{i} & =\left[\begin{array}{c}
\hat{{\rm p}}_{i}\\
1
\end{array}\right]-\left[\begin{array}{cc}
\mathcal{R}_{\hat{Q}} & \hat{P}\\
\underline{\mathbf{0}}_{3}^{\top} & 1
\end{array}\right]\left[\begin{array}{c}
y_{i}\\
1
\end{array}\right]=\left[\begin{array}{c}
e_{i}\\
0
\end{array}\right]\\
& \hspace{9em},\hspace{1em}i=1,2,\ldots,n\\
E_{i,k} & =\frac{1}{2}\text{ln}\frac{\underline{\delta}_{i,k}+e_{i,k}/\xi_{i,k}}{\bar{\delta}_{i,k}-e_{i,k}/\xi_{i,k}},\hspace{1em}k=1,2,3\\
\overline{{\rm Ad}}_{\hat{\boldsymbol{T}}}^{\top} & =\left[\begin{array}{cc}
\mathcal{R}_{\hat{Q}}^{\top} & -\mathcal{R}_{\hat{Q}}^{\top}\left[\hat{P}\right]_{\times}\\
0_{3\times3} & \mathcal{R}_{\hat{Q}}^{\top}
\end{array}\right]\\
\overline{{\rm Ad}}_{\hat{\boldsymbol{T}}^{-1}} & =\left[\begin{array}{cc}
\mathcal{R}_{\hat{Q}}^{\top} & 0_{3\times3}\\
-\mathcal{R}_{\hat{Q}}^{\top}\left[\hat{P}\right]_{\times} & \mathcal{R}_{\hat{Q}}^{\top}
\end{array}\right]\\
\chi & =\Omega_{m}-\hat{b}_{\Omega}-W_{\Omega}\\
\dot{\hat{Q}} & =\frac{1}{2}\left[\begin{array}{cc}
0 & -\chi^{\top}\\
\chi & -\left[\chi\right]_{\times}
\end{array}\right]\hat{Q},\hspace{1em}\hat{Q}(0)=Q_{{\rm I}}\\
\dot{\hat{P}} & =\mathbf{Y}\left(\hat{Q},V_{m}-\hat{b}_{V}-W_{V}\right)\\
\dot{{\rm \hat{p}}}_{i} & =-k_{p}\left(\Lambda_{i}+\Lambda_{i}^{-1}\right)E_{i}\\
\dot{\hat{b}}_{U} & =-\sum_{i=1}^{n}\frac{\Gamma}{\alpha_{i}}\overline{{\rm Ad}}_{\hat{\boldsymbol{T}}}^{\top}\left[\begin{array}{c}
\left[\mathbf{Y}(\hat{Q},y_{i})+\hat{P}\right]_{\times}\\
\mathbf{I}_{3}
\end{array}\right]\Lambda_{i}E_{i}\\
W_{U} & =-\sum_{i=1}^{n}k_{w}\overline{{\rm Ad}}_{\hat{\boldsymbol{T}}^{-1}}\left[\begin{array}{c}
\left[\mathbf{Y}(\hat{Q},y_{i})+\hat{P}\right]_{\times}\\
\mathbf{I}_{3}
\end{array}\right]\Lambda_{i}E_{i}
\end{cases}
\]

\end{document}